\documentclass[onecolumn,12pt,lettersize]{IEEEtran}
\usepackage[top=1 in, bottom=1 in, left=0.75 in, right=0.75 in]{geometry}

\usepackage[cmex10]{amsmath} 
\usepackage{amssymb}
\usepackage{tabularx}
\usepackage[sc,osf,noBBpl]{mathpazo} 
\usepackage{cite} 

\usepackage{graphicx} 
\usepackage{subcaption}
\usepackage{algorithm}
\usepackage[noend]{algorithmic}
\usepackage{amssymb} 
\usepackage{graphicx}
\usepackage{epstopdf}

\usepackage{stfloats}
\usepackage{tikz}

\usepackage{amsthm} 
\newtheorem{theorem}{Theorem}
\newtheorem{lemma}{Lemma}
\newtheorem*{lemma*}{Lemma}
\newtheorem*{claim*}{Claim}

\newtheorem{cor}{Corollary}
\theoremstyle{definition}

\newtheorem{assumption}{Assumption}
\theoremstyle{remark}
\newtheorem*{remark}{Remark}

\usepackage{flushend}

\begin{document}
\title{On the Rate of Convergence of Mean-Field Models: Stein's Method Meets the Perturbation Theory}
\author{\IEEEauthorblockN{Lei Ying}\\
\IEEEauthorblockA{School of Electrical, Computer and Energy Engineering\\
Arizona State University, Tempe, AZ 85287\\
lei.ying.2@asu.edu}}
\maketitle
\begin{abstract}
This paper studies the rate of convergence of a family of continuous-time Markov chains (CTMC) to a mean-field model. When the mean-field model is a finite-dimensional dynamical system with a unique equilibrium point, an analysis based on Stein's method and the perturbation theory shows that under some mild conditions, the stationary distributions of CTMCs converge (in the mean-square sense) to the equilibrium point of the mean-field model if the mean-field model is {\em globally asymptotically stable} and {\em locally exponentially stable}. In particular, the mean square difference between the $M$th CTMC in the steady state and the equilibrium point of the mean-field system is $O(1/M),$ where $M$ is the size of the $M$th CTMC.  This approach based on Stein's method provides a new framework for studying the convergence of CTMCs to their mean-field limit by mainly looking into the stability of the mean-field model, which is a deterministic system and is often easier to analyze than the CTMCs.  More importantly, this approach quantifies the rate of convergence, which reveals the approximation error of using mean-field models for approximating finite-size systems. \end{abstract}

\section{Introduction}
The mean-field method  is to study large-scale and complex stochastic systems using simple deterministic models. The idea of the mean-field method is to assume the states of nodes in a large-scale system are independently and identically distributed (i.i.d.). Based on this i.i.d. assumption, in a large-scale system, the interaction of a node to the rest of the system can be replaced with an ``average'' interaction, and the evolution of the system can then be modeled as a deterministic dynamical system, called a {\em mean-field model}. Then the macroscopic behaviors of the stochastic system can be approximated using the mean-field model, e.g., the stationary distribution of the stochastic system may be approximated using the equilibrium point of the mean-field model. The mean-field method has important applications in various areas including statistical physics, epidemiology, queueing theory, and game theory. Here are just a few examples of these applications \cite{Kad_09,Bai_75,BarKarKel_92,VveDobKar_96,Mit_96,LasLio_07}.

This paper focuses on the use of mean-field models for stationary distribution approximation of CTMCs, which has been used in queueing networks and recently in cloud computing systems for quantifying the performance of large-scale communication and computing systems such as data centers.  Besides solving the equilibrium point of the mean-field model, a critical component of the mean-field method in this application is to prove that the stationary distributions of a family of CTMCs indeed converge to the equilibrium point of the mean-field model, which justifies the use of the mean-field model.

In this paper, we consider a family of CTMCs, where the $M$th CTMC is an $M$-dimensional continuous-time Markov chain ${\bf W}^{(M)}\in {\cal U}^M,$ where the superscript $M$ is the number of nodes (or called particles) in the system and ${\cal U}^M\subseteq {\bf R}^n$ is the state space of the CTMC. We assume ${\cal U}$ is a finite state space and the CTMC is irreducible. Without loss of generality, let ${\cal U}=\{1, \cdots, n\}.$
We further define $$x^{(M)}_i(t)=\frac{1}{M}\sum_{m=1}^M {\mathbf 1}_{W_m^{(M)}(t)=i}$$ where $\mathbf 1$ is the indicator function, so $x^{(M)}_i(t)$ is the fraction of nodes in state $i$ at time $t.$ This paper focuses on the case such that ${\bf x}^{(M)}=\left\{{\bf x}^{(M)}(t), t\geq 0\right\}$ is also an ($n$-dimensional) CTMC, i.e., the CTMC is  a population process \cite{Kur_71,Kur_81}. We remark that many applications of the mean-field method such as those in queueing networks and epidemiology are for population processes.

Now let ${\bf x}^{(M)}(\infty)$ denote the stationary distribution of the $M$th CTMC. Furthermore, let ${\bf x}(t)$ denote the solution of an associated mean-field model and ${\bf x}^*$ denote its equilibrium point. Existing approaches for proving the convergence of ${\bf x}^{(M)}(\infty)$ to  ${\bf x}^*$ often involve the following three components.
\begin{itemize}
\item[(1)] The first component is to show the convergence of CTMCs to the trajectory of the mean-field model for any finite time interval $[0,t],$ i.e.,
\begin{equation}\lim_{M\rightarrow\infty} \sup_{0\leq s\leq t}d({\bf x}^{(M)}(s),{\bf x}(s))=0,\label{eq:finitetime}\end{equation} where $d(\cdot,\cdot)$ is some measure of distance. This can be proved using different techniques including Kurtz's theorem \cite{Kur_71,Kur_81,Mit_96,YinSriKan_15}, propagation of chaos \cite{Szn_91,AnaBen_93,BraLuPra_12}, or the convergence of the transition semigroup of CTMCs \cite{VveDobKar_96,MukMaz_13}.

\item[(2)] The second component is to prove the asymptotic stability of the mean-field model, i.e., $$\lim_{t\rightarrow\infty} {\bf x}(t)={\bf x}^*.$$ Lyapunov theorem or LaSalle invariance principle can often be used for proving the stability.

\item[(3)] After establishing the previous two results, we obtained
$$\lim_{t\rightarrow\infty} \lim_{M\rightarrow\infty}  {\bf x}^{(M)}(t)=\lim_{t\rightarrow\infty} {\bf x}(t)={\bf x}^*.$$ The convergence of the stationary distributions can then be concluded if we can prove the interchange of the limits, i.e., $$\lim_{M\rightarrow\infty}{\bf x}^{(M)}(\infty)= \lim_{M\rightarrow\infty}\lim_{t\rightarrow\infty} {\bf x}^{(M)}(t)=_{(a)}\lim_{t\rightarrow\infty} \lim_{M\rightarrow\infty}  {\bf x}^{(M)}(t) ={\bf x}^*,$$ where step $(a)$ is called the interchange of the limits.
\end{itemize}

Since these approaches are all based on the interchange of the limits and use the finite-time convergence (equality (\ref{eq:finitetime})) as the stepping stone, they are {\em indirect} methods of proving $$\lim_{M\rightarrow\infty}{\bf x}^{(M)}(\infty)={\bf x}^*.$$  Because of this reason, it is difficult to use these approaches to establish the rate of convergence of stationary distributions, and to provide bounds on the approximation error $\|{\bf x}^{(M)}(\infty)-{\bf x}^*\|$ for the finite-size system (i.e., for any fixed $M$).

This paper directly studies the convergence of the stationary distributions of CTMCs using Stein's method \cite{Ste_72,Ste_86}, which is a method to bound the distance of two probability distributions. Our use of Stein's method for the rate of convergence was inspired by the work by Braverman and Dai \cite{BraDai_15}, in which they developed a modular framework with three components for steady-state diffusion approximations and established the rate of convergence to diffusion models for $M/Ph/n+M$ queuing systems. The results in this paper also share similar spirit with the work by Gurvich \cite{Gur_14}, which establishes the rate of convergence of diffusion models for steady-state approximations for exponentially ergodic Markovian queues.

This paper focuses on mean-field models instead of diffusion models. Based on Stein's method, the paper identifies a fundamental connection between the perturbation theory for nonlinear systems and the convergence of mean-field models. The perturbation theory shows that for a stable nonlinear system with exponentially stable equilibrium point, the error of the first-order approximation of the nonlinear system is at the order of $O(\epsilon^2),$ where $\epsilon$ is the scaling factor of the perturbation. It turns out the mean-square difference between the stationary distribution of the $M$th CTMC and the equilibrium point of the mean-field model is related to the {\em cumulative} error (integrated over infinite time horizon) of using the first-order approximation of the mean-field model. After quantifying the cumulative error, the following results are established for finite-dimensional mean-field models.
\begin{itemize}
\item If the mean-field model is perfect (see definition in Section \ref{sec:stein}),  globally asymptotically stable and locally exponentially stable, then the stationary distributions of the CTMCs converges in the mean-square sense to the equilibrium point of the mean-field model with rate $O(1/M)$ (Theorem \ref{thm:main}), specifically,
\begin{equation}
E\left[\|{\bf x}^{(M)}(\infty)-{\bf x}^*\|^2\right]=O(\frac{1}{M}).\label{eq:rfc}
\end{equation}

\item If the mean-field model is not perfect, sufficient conditions that guarantee the convergence of the stationary distributions have been obtained in Corollary \ref{cor}.
\end{itemize}

We remark that these results are different from the celebrated law of large numbers for Markov chains established by Kurtz \cite{Kur_71,Kur_81}, where the convergence is established for sample paths of the CTMCs over a finite time interval, not for the stationary distributions of the CTMCs. The contributions of these results are two-fold: First, it provides a {\em direct} method of studying the convergence of stationary distributions to its mean-field limit. The method connects the convergence of CTMCs with the stability of the mean-field model. Note that the mean-field model is a deterministic system, so it is often easier to analyze than the CTMCs.  Second, the method quantifies the rate of convergence, and provides bounds on the approximation error when using the mean-field limit for approximating finite-size systems. Finally, recall that $$x^{(M)}_i(t)=\frac{1}{M}\sum_{m=1}^M {\mathbf 1}_{W_m^{(M)}(t)=i}$$ which is the average of $M$ Bernoulli random variables. The mean-square error in (\ref{eq:rfc}) is at the same order as the variance of the average of $M$ i.i.d. Beroulli random variables (according to law of large numbers). While it is obvious that ${\mathbf 1}_{W_m^{(M)}(t)=i}$ are not independent in the $M$th CTMC, the bound on the mean-square difference, however, provides an intuitive support to approximating a large-scale CTMC based on the i.i.d. assumption and mean-field models.

\section{Mean-Field Models}
\label{sec:stein}
Consider an $M$-dimensional continuous-time Markov chain ${\bf W}^{(M)}\in {\cal U}^M,$ where the superscript $M$ is the number of nodes (or called particles) in the system and ${\cal U}^M\subseteq {\bf R}^M$ is the state space of the CTMC. We assume ${\cal U}$ is a finite state space and the CTMC is irreducible. Without loss of generality, we assume ${\cal U}=\{1, \cdots, n\}.$
We further define $$X^{(M)}_i(t)=\sum_{m=1}^M {\mathbf 1}_{W_m^{(M)}(t)=i}$$ where $\mathbf 1$ is the indicator function, so $X^{(M)}_i(t)$ is the number of nodes in state $i$ at time $t.$ We further define 
$${\bf x}^{(M)}(t)=\frac{{\bf X}^{(M)}(t)}{M},$$ so $x^{(M)}_i(t)\in[0, 1]$ represents the {\em fraction} of nodes in state $i$ at time $t.$ In this paper, we assume ${\bf x}^{(M)}=\left\{{\bf x}^{(M)}(t), t\geq 0\right\}$ is an ($n$-dimensional) CTMC. We use ${\bf x}^{(M)}(\infty)$ to denote its stationary distribution.

Furthermore, we have a mean-field model described by the following autonomous dynamical system:
\begin{equation}\dot{\bf x}\triangleq \frac{d}{dt}x(t)=f({\bf x}(t))\quad {\bf x}(0)={\bf x}\hbox{ and } {\bf x}(t)\in{\cal D}\subseteq {[0, 1]}^n,\label{ori-sys-d}\end{equation} where $\cal D$ is a compact set. Here, we abuse the notation and use $\bf x$ to denote the initial condition, which simplifies the notation in the analysis later without causing too much confusion. Assume the system has a unique equilibrium point and let ${\bf x}^*$ denote the  equilibrium point. The key idea of the mean-field analysis is to use the solution of this deterministic dynamical system to approximate the behavior of the CTMC when $M$ is large; for example, use ${\bf x}^*$ to approximate ${\bf x}^{(M)}(\infty).$

Let $Q_{{\bf x}^{(M)}, {\bf y}^{(M)}}$ denote the transition rate of the CTMC from state ${\bf x}^M$ to state ${\bf y}^M.$ The sequence of CTMCs is called a density-dependent family of CTMCs if the normalized transition rate $q_{{\bf x}^{(M)}, {\bf y}^{(M)}}=\frac{1}{M}Q_{{\bf x}^{(M)}, {\bf y}^{(M)}}$ only depends on ${\bf x}^{(M)}$ and ${\bf y}^{(M)}$ but is independent of $M$ (see a detailed definition in \cite{Mit_96}). For a density-dependent family of CTMCs, the mean-field model can often be obtained by choosing $$f({\bf x})=\sum_{{\bf y}: {\bf y}\not={\bf x}} q_{{\bf x}, {\bf y}}\left({\bf y}-{\bf x}\right)$$ because $q_{{\bf x}, {\bf y}}$ is the transition rate from $\bf x$ to $\bf y$ and ${\bf y}-{\bf x}$ is the change of system state when such a transition occurs.

We next illustrate the idea using an SIS (susceptible-infected-susceptible) model with external infection source, which is a variation of the SIS model.

{\bf Example:} Let $W^{(M)}_m$ denote the state of an individual such that $W^{(M)}_m=0$ if the individual is susceptible and $W^{(M)}_m=1$ if the individual is infected. So $x^{(M)}_0$ is the fraction of susceptible individuals and $x^{(M)}_1$ is the fraction of infected individuals. We assume the recover time of an individual follows an exponential distribution with mean $1.$ Each infected node randomly selects an susceptible node after waiting for a random time that is exponentially distributed with mean $1/\beta$ and infects it. Each susceptible node, after it becomes susceptible, gets infected by an external infection source after a random time period that is exponentially distributed with mean $1/\alpha.$ Therefore,  $W^{(M)},$ ${\bf X}^{(m)}$ and ${\bf x}^{(M)}$ are CTMCs. Specifically, ${\bf x}^{(M)}$ has the following transition rates:
\begin{eqnarray*}
Q_{{\bf x}^{(M)}, {\bf y}^{(M)}}=\left\{
                                   \begin{array}{llll}
                                     M\alpha x_0^{(M)}+M\beta x_0^{(M)}x_1^{(M)}, & \hbox{ if } {\bf y}^{(M)}={\bf x}^{(M)} +\frac{1}{M}\left(
                                                 \begin{array}{c}
                                                   -1 \\
                                                   1 \\
                                                 \end{array}\right)\\
                                     Mx_1^{(M)}, & \hbox{ if } {\bf y}^{(M)}={\bf x}^{(M)} +\frac{1}{M}\left(
                                                 \begin{array}{c}
                                                   1 \\
                                                   -1 \\
                                                 \end{array}\right)\\
                                    -M\alpha x_0^{(M)}-M\beta x_0^{(M)}x_1^{(M)}-Mx_1^{(M)}, &\hbox{ if } {\bf y}^{(M)}={\bf x}^{(M)}\\
                                    0 & \hbox{ otherwise.}
                                   \end{array}
                                 \right.
\end{eqnarray*}

Note for a given $M,$ computing the stationary distribution of ${\bf x}^{(M)}$ is not easy because it has a large state space $$\left\{\frac{1}{M}, \frac{2}{M},\cdots, 1\right\}^2$$ and the transition rates are nonlinear functions of the state.

The SIS considered above is a density-dependent family, so we consider the following mean-field model
\begin{equation}
\left(
                                                 \begin{array}{c}
                                                   \dot{x}_0 \\
                                                   \dot{x}_1 \\
                                                 \end{array}
\right)=f({\bf x})=\sum_{{\bf y}: {\bf y}\not={\bf x}} q_{{\bf x}, {\bf y}}\left({\bf y}-{\bf x}\right)=(\alpha x_0+\beta x_0 x_1)\left(
                                                 \begin{array}{c}
                                                   -1 \\
                                                   1 \\
                                                 \end{array}
\right)+x_1\left(
                                                 \begin{array}{c}
                                                   1 \\
                                                   -1 \\
                                                 \end{array}
\right)
\end{equation}
To solve the mean-field model above, we notice that $x_0+x_1=1$ always holds, so we only need to consider
\begin{eqnarray*}
\dot{x}_0=-\alpha x_0-\beta x_0 (1-x_0)+(1-x_0).
\end{eqnarray*}
The equilibrium point can then be obtained by solving
\begin{eqnarray*}
0=-\alpha x_0-\beta x_0 (1-x_0)+(1-x_0).
\end{eqnarray*}
For example, if $\alpha=\beta=0.5,$ then
\begin{eqnarray*}
x^*_0=2-\sqrt{2} \quad\hbox{and}\quad x^*_1=\sqrt{2}-1,
\end{eqnarray*}  which can be used to approximate the fractions of susceptible and infected populations when $M$ is large, i.e., the stationary distribution of ${\bf x}^{(M)}.$ The simulation results of the fraction of infected population with $M=100,$ $1,000$ and $100,000$ are shown in Figure \ref{fig:sis}, from which the convergence of $x_0^{(M)}$ to $2-\sqrt{2}$ can be observed.

\begin{figure}[htb]
        \centering
 		\includegraphics[width=0.6\columnwidth]{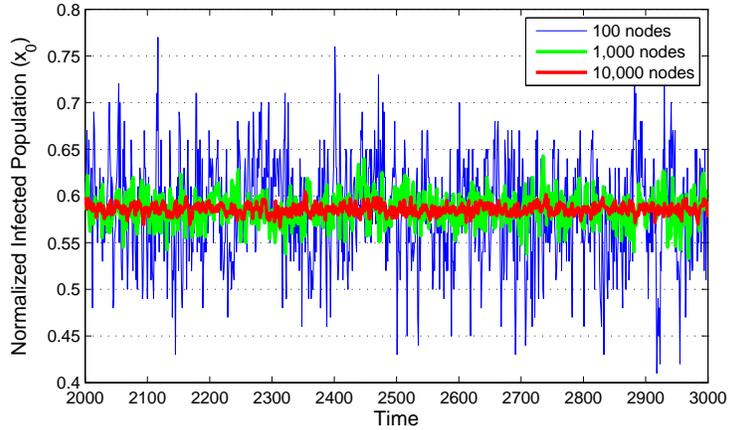}
        \caption{Simulation results of the normalized infected population with $M=100,$ $1,000,$ and $100,000.$ $\alpha=\beta=0.5$ in these simulations.  The CTMCs were simulated using the uniformization method. The time is the discrete-time (scaled with $M$).}\label{fig:sis}
\end{figure}

\hfill{$\square$}

\section{Stein's Method for the Rate of Convergence}

In this section, we study the convergence of the CTMCs to a mean-field model using Stein's method and the perturbation theory. Throughout this paper, $\|\cdot\|$ denotes the 2-norm, i.e., $\|{\bf x}\|=\sqrt{\sum_i x_i^2},$ and $|\cdot|$ denotes the absolute value. For two vectors ${\bf a}, {\bf b} \in {\bf R}^n,$ ${\bf a} \cdot {\bf b}$ is the dot product. Furthermore, $\triangledown g({\bf x})$ denotes the gradient of $g({\bf x}),$ and $\triangledown x_i(t,{\bf x})$ refers to differentiating with respect to the location $\bf x$, and $\dot{\bf x}$ is the derivative with respect to time.

Recall the mean-field model defined in (\ref{ori-sys-d}): 
\begin{equation*}\dot{\bf x}=f({\bf x}(t))\quad {\bf x}(0)={\bf x} \hbox{ and } {\bf x}(t)\in{\cal D}\subseteq [0,1]^n.\end{equation*} The mean-field model is said to be {\em globally asymptotically stable} if given any initial condition ${\bf x}(0)\in{\cal D}$ and any $\epsilon>0,$ there exists $t({\bf x}(0), \epsilon)$ such that $$\|{\bf x}(t)-{\bf x}^*\|\leq \epsilon \quad \forall t\geq t({\bf x}(0), \epsilon).$$ The mean-field model is said to be {\em locally exponentially stable} if there exist positive constants $\epsilon,$ $\alpha$ and $\kappa$ such that starting from any initial condition $\|{\bf x}(0)\|\leq \epsilon,$ \begin{equation*}
\|{\bf x}(t)-{\bf x}^*\|\leq \kappa \|{\bf x}(0)\|\exp\left(-\alpha t\right).
\end{equation*}

Let $g({\bf x})$ be the solution to the Poisson equation
\begin{equation}
\triangledown g({\bf x})\cdot \dot{\bf x}=\triangledown g({\bf x})\cdot f({\bf x})=\sum_{i=1}^n \left(x_i-x_i^*\right)^2\label{eq:poisson}
\end{equation}
Then, $$g({\bf x})=-\int_0^\infty \sum_i \left(x_i(t, {\bf x})-x_i^*\right)^2\, dt$$ when the integral is finite \cite{Bar_88,Got_91}, where ${\bf x}(t, {\bf x})$ is the trajectory of the dynamical system with ${\bf x}$ as the initial condition. The integral is finite when the mean-field model is asymptotically stable and locally exponentially stable. Note that $-g({\bf x})$ can be viewed as the cumulative square-deviation of the system state from the equilibrium point when the initial condition is $\bf x.$

Now let $G_{{\bf x}^{(M)}}$ denote the generator for the $M$th CTMC, then
\begin{eqnarray*}
G_{{\bf x}^{(M)}}g({\bf x}) &=& \sum_{{\bf y}: {\bf y}\not={\bf x}}Q_{{\bf x},{\bf y}}({\bf x})\left(g({\bf y})-g({\bf x})\right)\\
&=&M\sum_{{\bf y}: {\bf y}\not={\bf x}}q_{{\bf x},{\bf y}}({\bf x})\left(g({\bf y})-g({\bf x})\right).
\end{eqnarray*}
Since ${\bf x}^{(M)}$ is irreducible and has finite state space, ${\bf x}^{(M)}$ has a stationary distribution.  Initializing ${\bf x}^{(M)}(0)$ according to its stationary distribution, and using  $E_{{\bf x}^{(M)}}[\cdot]$ throughout to denote stationary expectation, we have
\begin{eqnarray}
E_{{\bf x}^{(M)}}\left[G_{{\bf x}^{(M)}}g({\bf x})\right]
=E_{{\bf x}^{(M)}}\left[M\sum_{{\bf y}: {\bf y}\not={\bf x}}q_{{\bf x},{\bf y}}({\bf x})\left(g({\bf y})-g({\bf x})\right)\right]=0,\label{eq:ss}
\end{eqnarray} where the subscript in the expectation indicates that the expectation is taken over the stationary distribution of ${\bf x}^{(M)}.$
Then by taking expectation of the Poisson equation (\ref{eq:poisson}) over the stationary distribution of ${\bf x}^{(M)}$ and then adding (\ref{eq:ss}) to the equation, we obtain
\begin{eqnarray}
&&E_{{\bf x}^{(M)}}\left[\sum_{i=1}^n \left(x_i-x_i^*\right)^2\right]\nonumber\\
&=&E_{{\bf x}^{(M)}}\left[\triangledown g({\bf x})\cdot f({\bf x})-M\sum_{{\bf y}: {\bf y}\not={\bf x}}q_{{\bf x},{\bf y}}\left(g({\bf y})-g({\bf x})\right)\right]\nonumber\\
\end{eqnarray}
Now adding and subtracting $\triangledown g({\bf x})\cdot\sum_{{\bf y}: {\bf y}\not={\bf x}}q_{{\bf x},{\bf y}}M({\bf y}-{\bf x})$ yields
\begin{eqnarray}
&&E_{{\bf x}^{(M)}}\left[\sum_{i=1}^n \left(x_i-x_i^*\right)^2\right]\nonumber\\
&=&E_{{\bf x}^{(M)}}\left[\triangledown g({\bf x})\cdot f({\bf x})-\triangledown g({\bf x})\cdot\sum_{{\bf y}: {\bf y}\not={\bf x}}q_{{\bf x},{\bf y}}M({\bf y}-{\bf x})- M\sum_{{\bf y}: {\bf y}\not={\bf x}}q_{{\bf x},{\bf y}}\left(g({\bf y})-g({\bf x})-\triangledown g({\bf x})\cdot({\bf y}-{\bf x})\right)\right]\nonumber\\
&=&E_{{\bf x}^{(M)}}\left[\triangledown g({\bf x})\cdot \left(f({\bf x})-\sum_{{\bf y}: {\bf y}\not={\bf x}}q_{{\bf x},{\bf y}}M({\bf y}-{\bf x})\right)- \sum_{{\bf y}: {\bf y}\not={\bf x}}q_{{\bf x},{\bf y}}M\left(g({\bf y})-g({\bf x})-\triangledown g({\bf x})\cdot({\bf y}-{\bf x})\right)\right].\label{eq:stein}
\end{eqnarray}

From the equality above, {\em intuitively}, that $E_{{\bf x}^{(M)}}\left[\sum_{i=1}^n \left(x_i-x_i^*\right)^2\right]$ converges to zero as $M\rightarrow\infty$ can be established if the followings are true:
\begin{itemize}
\item Bounded gradient of $g({\bf x}):$ $\left\|\triangledown g({\bf x})\right\|$ is bounded by a constant independent of $M.$

\item Convergence of the generator: $$\lim_{M\rightarrow\infty}E_{{\bf x}^{(M)}}\left[\left\|f({\bf x})-\sum_{{\bf y}: {\bf y}\not={\bf x}}q_{{\bf x},{\bf y}}M({\bf y}-{\bf x})\right\|\right]=0.$$

\item Bounded transition-rate of the CTMC: $E_{{\bf x}^{(M)}}\left[ \sum_{{\bf y}: {\bf y}\not={\bf x}}q_{{\bf x},{\bf y}}\right]$ is bounded.

\item Diminishing first-order approximation error: $$\left\|g({\bf y})-g({\bf x})-\triangledown g({\bf x})\cdot({\bf y}-{\bf x})\right\|=O\left(\frac{1}{M^2}\right).$$ Note that $g({\bf x})+\triangledown g({\bf x})\cdot({\bf y}-{\bf x})$ is the first-order Taylor approximation of $g({\bf y}).$
\end{itemize}

For many CTMCs and the associated mean-field models, the first three conditions mentioned above can be easily verified. In the following theorem, we will prove that the last condition holds when the mean-field model is globally asymptotically stable and locally exponentially stable (see inequality (\ref{eq:order})), and then establish the rate of convergence based on that. The following theorem presents the main result of this paper. 

\begin{theorem} The stationary distributions of the CTMCs (${\bf x}^{(M)}(\infty)$), defined in Section \ref{sec:stein}, converge to the equilibrium point (${\bf x}^*$) of the mean-field model (\ref{ori-sys-d})  in the mean-square sense with rate $1/M,$ i.e.,
$$E_{{\bf x}^{(M)}}\left[\sum_{i=1}^n \left(x_i-x_i^*\right)^2\right]=O\left(\frac{1}{M}\right)$$ when the following conditions hold:
\begin{itemize}
\item {\bf Bounded transition-rate condition:} There exists a constant $c>0$ independent of $M$ such that $$E_{{\bf x}^{(M)}}\left[ \sum_{{\bf y}: {\bf y}\not={\bf x}}q_{{\bf x},{\bf y}}\right]\leq c.$$

\item {\bf Bounded state transition condition:} There exists a constant $\tilde{c}$ independent of $M$ such that $\|{\bf x}-{\bf y}\|\leq \frac{\tilde{c}}{M}$ for any $\bf x$ and $\bf{y}$ such that $q_{{\bf x}, {\bf y}}\not=0.$

\item {\bf Perfect mean-field model condition:} The mean-field model (\ref{ori-sys-d}) is given by $$f({\bf x})=\sum_{{\bf y}: {\bf y}\not={\bf x}}q_{{\bf x},{\bf y}}M({\bf y}-{\bf x})  \quad \forall\ {\bf x}.$$

\item {\bf Partial derivative condition:} The function $f({\bf x})$ is twice continuously differentiable. 

\item {\bf Stability condition:} The mean-field model is globally asymptotically stable and is locally exponentially stable.
\end{itemize}
\label{thm:main}
\end{theorem}
\begin{remark}
The first four conditions are easy to verify, so only the stability condition requires nontrivial work. Since a dynamical system has an exponentially stable equilibrium point if and only if the linearized system (at the equilibrium) is exponentially stable (see Theorem 4.15 in \cite{Kha_01}),  the local exponential stability can be verified by calculating the eigenvalues of the state matrix of the mean-field model. When the parameters of the mean-field model are given, the local exponential stability can be numerically verified. The global asymptotical stability in general is studied using the Lyapunov theorem.
\end{remark}

\begin{remark}
It is worth to pointing out that if the mean-field model is unstable but the perfect mean-field model assumption holds. Kurtz's theorem indicates that the sample paths of the CTMCs converge to the trajectory of the mean-field model for any finite time interval, which implies that the CTMCs are ``unstable'' as well.
\end{remark}

\begin{proof} We first prove the theorem assuming the mean-field model is globally exponentially stable, and then extend it to the general case.
Under the perfect mean-field model assumption, equation (\ref{eq:stein}) becomes
\begin{eqnarray*}
&&E_{{\bf x}^{(M)}}\left[\sum_{i=1}^n \left(x_i-x_i^*\right)^2\right]\\
&=&E_{{\bf x}^{(M)}}\left[- \sum_{{\bf y}: {\bf y}\not={\bf x}}q_{{\bf x},{\bf y}}M\left(g({\bf y})-g({\bf x})-\triangledown g({\bf x})\cdot({\bf y}-{\bf x})\right)\right],
\end{eqnarray*}
where $$g({\bf x})=-\int_0^\infty \sum_i \left(x_i(t, {\bf x})-x_i^*\right)^2\, dt.$$  We next focus on the following term,
\begin{eqnarray}
&&-\left(g({\bf y})-g({\bf x})-\triangledown g({\bf x})\cdot({\bf y}-{\bf x})\right)\\
&=&\left.\int_0^\infty \sum_i \left(\left(x_i(t, {\bf y})-x_i^*\right)^2-\left(x_i(t, {\bf x})-x_i^*\right)^2-2\left(x_i(t, {\bf x})-x_i^*\right)\triangledown x_i(t, {\bf x})\cdot({\bf y}-{\bf x})\right)\,dt\right.\label{eq:2nd-deri}
\end{eqnarray}
Note that we exchanged the order of integration and differentiation for the third term. This is can be done because
$$\int_0^\infty 2\left(x_i(t, {\bf x})-x_i^*\right)\triangledown x_i(t, {\bf x})\cdot({\bf y}-{\bf x})\,dt$$ is finite, which can be proved using the fact that both $\left(x_i(t, {\bf x})-x_i^*\right)$ and $\|\triangledown x_i(t, {\bf x})\|$ decay exponentially fast to zero as $t$ increases (apply inequalities (\ref{eq:exconv}) and (\ref{eq:exconv-1}) with ${\bf z}={\bf 1}$), and the fact that $\|{\bf y}-{\bf x}\|$ is bounded due to the bounded state transition condition.

We next define \begin{equation*}
e_i(t)=x_i(t, {\bf y})-x_i(t, {\bf x})-\triangledown x_i(t, {\bf x})\cdot({\bf y}-{\bf x}),
\end{equation*} i.e.,  
\begin{equation*}
x_i(t, {\bf y})=e_i(t)+x_i(t, {\bf x})+\triangledown x_i(t, {\bf x})\cdot({\bf y}-{\bf x}),
\end{equation*}
so
\begin{eqnarray*}
&&\left(x_i(t, {\bf y})-x_i^*\right)^2-\left(x_i(t, {\bf x})-x_i^*\right)^2-2\left(x_i(t, {\bf x})-x_i^*\right)\triangledown x_i(t, {\bf x})\cdot({\bf y}-{\bf x})\\
&=&\left(e_i(t)+x_i(t, {\bf x})-x_i^*+\triangledown x_i(t, {\bf x})\cdot({\bf y}-{\bf x})\right)^2-\left(x_i(t, {\bf x})-x_i^*\right)^2-2\left(x_i(t, {\bf x})-x_i^*\right)\triangledown x_i(t, {\bf x})\cdot({\bf y}-{\bf x})\\
&=&e^2_i(t)+\left(\triangledown x_i(t, {\bf x})\cdot({\bf y}-{\bf x})\right)^2+2e_i(t)\triangledown x_i(t, {\bf x})\cdot({\bf y}-{\bf x})+2e_i(t)\left(x_i(t, {\bf x})-x_i^*\right)\\
&=&e_i(t)\left(e_i(t)+2\triangledown x_i(t, {\bf x})\cdot({\bf y}-{\bf x})+2\left(x_i(t, {\bf x})-x_i^*\right)\right)+\left(\triangledown x_i(t, {\bf x})\cdot({\bf y}-{\bf x})\right)^2.
\end{eqnarray*}

According to the perturbation theory, in particular, inequality (\ref{eq:error}), when the system is exponentially stable, we have that $$|e_i(t)|\leq \|{\bf e}(t)\|=O\left(\frac{1}{M^2}\right).$$ According to the bounded state transition condition, $$\|{\bf x}-{\bf y}\|\leq \frac{\tilde{c}}{M}.$$ Furthermore, both  $\triangledown x_i(t, {\bf x})$ and
$x_i(t, {\bf x})$ are bounded (see inequalities (\ref{eq:exconv}) and (\ref{eq:exconv-1})) by constants independent of $M$ and $t.$ Therefore, we can choose  a constant $b$ and a sufficiently large $\tilde{M}$ such that for any $M\geq \tilde{M},$
$$\left|e_i(t)+2\triangledown x_i(t, {\bf x})\cdot({\bf y}-{\bf x})+2\left(x_i(t, {\bf x})-x_i^*\right)\right|\leq b,$$ which implies that
\begin{eqnarray}
\left|g({\bf y})-g({\bf x})-\triangledown g({\bf x})\cdot({\bf y}-{\bf x})\right|
&\leq& b{\int}_0^\infty \sum_i|{e}_i(t)|\,dt+\int_0^\infty \sum_i \left(\triangledown x_i(t, {\bf x})\cdot({\bf y}-{\bf x})\right)^2\,dt,\\
&\leq&  b\sqrt{n} {\int}_0^\infty \|{\bf e}(t)\|\,dt+\int_0^\infty \sum_i \left(\triangledown x_i(t, {\bf x})\cdot({\bf y}-{\bf x})\right)^2\,dt,
\label{eq:boundong} 
\end{eqnarray}  where the last inequality is based on the following relation between 1-norm and 2-norm: $\sum_i|e_i(t)|\leq \sqrt{n}\|{\bf e}(t)\|.$ In Section \ref{sec:pert}, we will show that under the exponential stability assumption, $${\int}_0^\infty \|{\bf e}(t)\|\,dt=O(1/M^2) \quad \hbox{see inequality (\ref{eq:accumulated-error})}.$$

From the bounded state transition condition, $\|{\bf  y}-{\bf x}\|^2\leq \frac{\tilde{c}^2}{M^2}.$ Therefore, 
\begin{eqnarray*}
\int_0^\infty \sum_i \left(\triangledown x_i(t, {\bf x})\cdot({\bf y}-{\bf x})\right)^2\,dt\leq \frac{\tilde{c}^2}{M^2} \int_0^\infty \sum_i \left\|\triangledown x_i(t, {\bf x})\right\|^2\,dt.
\end{eqnarray*}
Now according to inequality (\ref{eq:exconv-1}), there exist positive constants $b_1$ and $b_2,$ both independent $M,$ such that
\begin{eqnarray*}
&&\int_0^\infty \sum_i \left(\triangledown x_i(t, {\bf x})\cdot({\bf y}-{\bf x})\right)^2\,dt\leq \frac{\tilde{c}^2}{M^2}\int_0^\infty b_1\exp(-b_2t)\,dt\leq \frac{b_1}{b_2}\frac{1}{M^2}.
\end{eqnarray*} Therefore, we can conclude that
\begin{eqnarray}
\left|g({\bf y})-g({\bf x})-\triangledown g({\bf x})\cdot({\bf y}-{\bf x})\right|=O\left(\frac{1}{M^2}\right),
\label{eq:order}
\end{eqnarray} which implies that
\begin{eqnarray}
E_{{\bf x}^{(M)}}\left[\sum_{i=1}^n \left(x_i-x_i^*\right)^2\right]=O\left(\frac{1}{M}\right) E_{{\bf x}^{(M)}}\left[ \sum_{{\bf y}: {\bf y}\not={\bf x}}q_{{\bf x},{\bf y}}\right].\label{eq:bound2}
\end{eqnarray}
Finally, using the bounded transition rate condition, we conclude
\begin{eqnarray}
E_{{\bf x}^{(M)}}\left[\sum_{i=1}^n \left(x_i-x_i^*\right)^2\right]=O\left(\frac{1}{M}\right).\label{eq:bound2}
\end{eqnarray}

Now consider the case that the mean-field model is not globally exponentially stable, but is globally asymptotically stable and locally exponentially stable. Recall that ${\cal D}\subseteq [0,1]^n$ is compact. According to the definition of global asymptotic stability (Definition 4.4 in \cite{Kha_01}), given any $\epsilon>0,$ there exists a finite time $t'$ such that $$\|{\bf x}(t)\|\leq \epsilon$$ for any $t\geq t'.$  For any finite $t,$ following a similar analysis as in Section \ref{sec:pert} (or Section 10.1 in \cite{Kha_01}), $\|{\bf e}(t, {\bf x})\|=O(1/M^2)$ holds\footnote{This holds without exponential stability, but the constant in $O(1/M^2)$ may be a function of $t$ if the system is not exponentially stable.} Therefore, we can bound the term (\ref{eq:2nd-deri}) by separating the integration into two intervals: from 0 to $t',$ and from $t'$ to $\infty,$ where $t'$ is chosen such that ${\bf x}(t)$ converges exponentially to the equilibrium point after $t'.$  Since $\|{\bf e}(t', {\bf x})\|=O(1/M^2),$ the analysis above applies to the integration over $[t',\infty).$ Hence, the result holds.
\end{proof}

The theorem above requires a {\em perfect mean-field model} and {\em bounded state transitions}.  Both conditions can be relaxed, but the rate of convergence will be different.
\begin{cor}
Assume partial derivative condition and the stability condition in Theorem \ref{thm:main} hold.  The stationary distributions of the CTMCs converge (in the mean square sense) to equilibrium point of the mean-field model, i.e.,
$$\lim_{M\rightarrow\infty} E_{{\bf x}^{(M)}}\left[\sum_{i=1}^n \left(x_i-x_i^*\right)^2\right]=0$$ when the following conditions also hold:
\begin{eqnarray}
\lim_{M\rightarrow\infty}E_{{\bf x}^{(M)}}\left[\left\|f({\bf x})-\sum_{{\bf y}: {\bf y}\not={\bf x}}q_{{\bf x},{\bf y}}M({\bf y}-{\bf x})\right\|\right]=0.\label{eq:add-1}\\
\lim_{M\rightarrow\infty}E_{{\bf x}^{(M)}}\left[\sum_{{\bf y}: {\bf y}\not={\bf x}}q_{{\bf x},{\bf y}}M\|{\bf y}-{\bf x}\|^2\right]=0\label{eq:add-2}\\
\lim_{M\rightarrow\infty}\max_{{\bf x}, {\bf y}: q_{{\bf x}, {\bf y}}>0}\|{\bf y}-{\bf x}\|=0\label{eq:add-3}
\end{eqnarray}
We say that the mean-field model is asymptotically accurate when condition (\ref{eq:add-1}) holds, which replaces the perfect mean-field model condition. Conditions (\ref{eq:add-2}) and (\ref{eq:add-3}) replace the bounded state transition condition. \label{cor}
\end{cor}
\begin{proof} First recall that we have
\begin{eqnarray}
&&E_{{\bf x}^{(M)}}\left[\sum_{i=1}^n \left(x_i-x_i^*\right)^2\right]\nonumber\\
&=&E_{{\bf x}^{(M)}}\left[\triangledown g({\bf x})\cdot \left(f({\bf x})-\sum_{{\bf y}: {\bf y}\not={\bf x}}q_{{\bf x},{\bf y}}M({\bf y}-{\bf x})\right)- \sum_{{\bf y}: {\bf y}\not={\bf x}}q_{{\bf x},{\bf y}}M\left(g({\bf y})-g({\bf x})-\triangledown g({\bf x})\cdot({\bf y}-{\bf x})\right)\right]\nonumber\\
&\leq& \left(\max_{{\bf x}} \left\|\triangledown g({\bf x})\right\|\right) E_{{\bf x}^{(M)}}\left[\left\|f({\bf x})-\sum_{{\bf y}: {\bf y}\not={\bf x}}q_{{\bf x},{\bf y}}M({\bf y}-{\bf x})\right\|\right]+ \label{eq:gen}\\
&&E_{{\bf x}^{(M)}}\left[\sum_{{\bf y}: {\bf y}\not={\bf x}}q_{{\bf x},{\bf y}}M\left|g({\bf y})-g({\bf x})-\triangledown g({\bf x})\cdot({\bf y}-{\bf x})\right|\right].\label{eq:pert}
\end{eqnarray}

By choosing ${\bf z}={\bf 1}$ in Section \ref{sec:pert}, it is easy to verify according to inequality (\ref{eq:exconv-1})  that $\left(\max_{{\bf x}} \left\|\triangledown g({\bf x})\right\|\right)$ is upper bounded by a constant independent of $M.$ Therefore, under condition (\ref{eq:add-1}), $(\ref{eq:gen})\rightarrow 0$ as $M\rightarrow \infty.$

A careful examination of inequality (\ref{eq:accumulated-error}) shows that
\begin{equation*}
\int_0^\infty \|{\bf e}(t)\|\,dt =O\left(\|{\bf y}-{\bf x}\|^2\exp\left(\alpha_3^\prime \|{\bf y}-{\bf x}\|\right)\right).
\end{equation*}  So under condition (\ref{eq:add-3}), we have
\begin{equation*}
\int_0^\infty \|{\bf e}(t)\|\,dt =O\left(\|{\bf y}-{\bf x}\|^2\right).
\end{equation*}
When condition (\ref{eq:add-3}) holds, following the analysis that leads to inequality (\ref{eq:boundong}), we can again show that there exists a constant $\tilde{b}$ independent $M$ such that
\begin{eqnarray*}
\left|g({\bf y})-g({\bf x})-\triangledown g({\bf x})\cdot({\bf y}-{\bf x})\right|
\leq  \tilde{b}\sqrt{n} {\int}_0^\infty \|{\bf e}(t)\|\,dt+\int_0^\infty \sum_i \left(\triangledown x_i(t, {\bf x})\cdot({\bf y}-{\bf x})\right)^2\,dt.
\end{eqnarray*}
According to inequality (\ref{eq:exconv-1}), we also have
\begin{eqnarray*}
\int_0^\infty \sum_i \left(\triangledown x_i(t, {\bf x})\cdot({\bf y}-{\bf x})\right)^2\,dt=O\left(\|{\bf y}-{\bf x}\|^2\right).
\end{eqnarray*}
Therefore, we have
\begin{eqnarray*}
(\ref{eq:pert})=O\left(E_{{\bf x}^{(M)}}\left[\sum_{{\bf y}: {\bf y}\not={\bf x}}q_{{\bf x},{\bf y}}M\|{\bf y}-{\bf x}\|^2\right]\right),
\end{eqnarray*} which converges to zero according to condition (\ref{eq:add-2}). Hence, the  corollary holds.
\end{proof}
\begin{remark}
When the mean-field model is asymptotically accurate, the convergence rate depends on the convergence rates of (\ref{eq:add-1}) and (\ref{eq:add-2}).
\end{remark}

{\bf Example:}  Let us first go back to the SIS model introduced in Section \ref{sec:stein}. A closed-form solution can be obtained for $x_0(t).$ Again assume $\alpha=\beta=0.5,$ then the solution of the ordinary differential equation is
$$x_0(t)=\frac{e^{-\sqrt{2}t}(2+\sqrt{2})\left(x_0(0)-2+\sqrt{2}\right)+(2-\sqrt{2})x_0(0)-2}{-e^{-\sqrt{2}t}\left(x_0(0)-2+\sqrt{2}\right)+x_0(0)-2-\sqrt{2}},$$ which converges to $2-\sqrt{2}$ as $t\rightarrow\infty$ independent of $x_0(0).$ Therefore, it is easy to verify that the system is globally, asymptotically stable. Furthermore, the linearized system at the equilibrium is
$$\dot{\epsilon}_0=-\left(\alpha+\beta(1-x_0^*)+1\right)\epsilon_0,$$ where $\epsilon_0=x_0-x^*_0$ and $x^*_0$ is the equilibrium value, so the equilibrium point is locally exponentially stable. Furthermore, the mean-field model is perfect in this case and it can be easily verified that all other conditions in Theorem \ref{thm:main} hold. So in the mean square sense, stationary distributions converge to the $x_0=2-\sqrt{2}$ and $x_1=\sqrt{2}-1$ with rate $1/M.$ Numerical evaluation of $ME_{{\bf x}^{(M)}}\left[\left(x_0-x_0^*\right)^2\right]$ versus $M$ is shown in Figure \ref{fig:msd}, where $M$ varies from 100 to 1,000. We can see that $ME_{{\bf x}^{(M)}}\left[\sum_{i=1}^n \left(x_i-x_i^*\right)^2\right]$ varies within the interval [0.21, 0.27] while the size of the system increases by 10 times (from 100 to 1,000). The standard deviation (deviation from $2-\sqrt{2}$) is 0.02177 (3.72\%$\times (2-\sqrt{2})$)) when $M=100$ and is 0.0068 (1.16\%$\times (2-\sqrt{2})$) when $M=1,000.$

\begin{figure}[htb]
        \centering
 		\includegraphics[width=0.6\columnwidth]{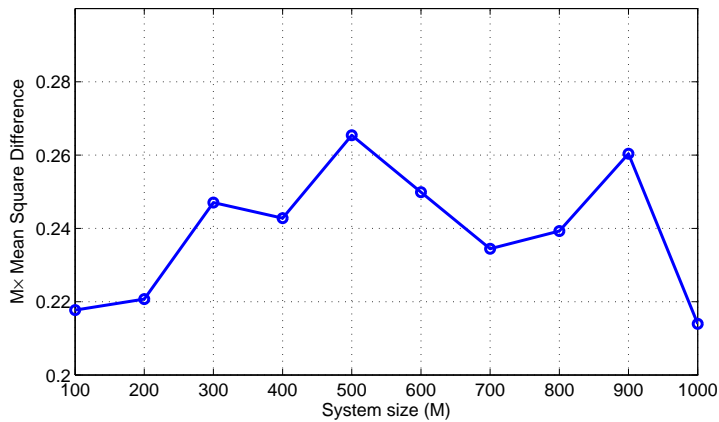}
        \caption{Numerical evaluation of $ME_{{\bf x}^{(M)}}\left[\left(x_0-x_0^*\right)^2\right]$ versus $M.$ }\label{fig:msd}
\end{figure}

\section{The Perturbation Theory}
\label{sec:pert}
In this section, we summarize the results of the perturbation theory for nonlinear systems. These results are special cases of those results presented in \cite{Kha_01} because we only need to consider a perturbation to the initial condition. Furthermore, the mean-field model considered in this paper is an autonomous system, which again is a special case of the nonlinear system considered in \cite{Kha_01}. For these reasons, the analysis of the perturbation results can be simplified. On the other hand, the perturbation method introduced in \cite{Kha_01} only states that the 2-norm of the following error is is at the order of $\|{\bf y}-{\bf x}\|^2$ independent of $t$ (under certain conditions) \begin{equation*}
{\bf e}(t)={\bf x}(t, {\bf y})-{\bf x}(t, {\bf x})-\triangledown {\bf x}(t, {\bf x})\cdot({\bf y}-{\bf x})
\end{equation*} Our result on the rate of convergence, however, requires such an upper bound on the cumulative error, i.e., an upper bound on
 \begin{equation*}
\int_0^\infty {\bf e}(t)\,dt.
\end{equation*} Therefore, it is necessary to go through the detailed analysis for the system considered in this paper to establish the result for the cumulative error.  For the completeness of the paper and the easy reference of the reader, we next introduce the perturbation results in \cite{Kha_01} with a more detailed calculation of $\|{\bf e}(t)\|,$ which shows that not only the approximation error is bounded, but the upper bound decays exponentially to zero as $t$ increases. The analysis closely follows \cite{Kha_01}.

Consider the system \begin{equation} \dot{{\bf x}}=f({\bf x})\label{ori-sys}\end{equation} where $f: {\cal D}\subseteq [0, 1]^n\rightarrow R^{n}.$  {\em Without the loss of generality, we assume ${\bf x}^*={\bf 0}.$} We are interested in comparing the solution of this nominal system with the system with a perturbation on the initial condition ${\bf x}(0)={\bf x}+\epsilon {\bf z},$ where ${\bf z}=\frac{1}{\epsilon}({\bf y}-{\bf x})$ and is an $n$-dimensional vector. For the mean-field analysis considered in this paper, $\epsilon=1/M.$ Under the condition of Theorem \ref{thm:main}, for any neighboring states $\bf x$ and $\bf y,$ $$\|{\bf z}\|=\frac{1}{\epsilon}\|({\bf y}-{\bf x})\|\leq \tilde{c}.$$

Let ${\bf x}(t, \epsilon)$ to denote the solution of the dynamical system with initial perturbation $\epsilon.$ Note that the dependence of the solution on ${\bf y}-{\bf x}$ is omitted to simplify the notation. The analysis holds for any ${\bf y}$ and ${\bf x}.$ We next first repeat the assumptions on the nominal dynamical system.

\begin{assumption} For any $i,$ the function $f_i({\bf x})$ is twice continuously differentiable. Therefore, the Jacobian matrix of $f({\bf x}),$ denoted by $\frac{\partial f}{\partial x},$ is Lipschitz. In other words, there exists a constant $L>0$ such that
\begin{eqnarray*}
\left\|\frac{\partial f}{\partial x}\left({\bf x}\right)-\frac{\partial f}{\partial x}\left({\bf y}\right)\right\|&\leq& L \|{\bf x}-{\bf y}\|.
\end{eqnarray*}
\hfill{$\square$}\label{asump:2-diff}
\end{assumption}

\begin{assumption} The dynamical system (\ref{ori-sys}) has a unique equilibrium point and is exponentially stable. In other words, there exist positive constants $\alpha$ and $\kappa$ such that starting from any initial condition ${\bf x}(0)\in {\cal D},$ \begin{equation}
\|{\bf x}(t)\|\leq \kappa \|{\bf x}(0)\|\exp\left(-\alpha t\right).\label{eq:exconv}
\end{equation}
Under this assumption, according to Theorem 4.14 in \cite{Kha_01}, there exist a Lyapunov function $V({\bf x})$ and positive constants $c_u,$ $c_l,$ $c_d,$ and $c_p$ such that for any ${\bf x}\in{\cal D},$ the following inequalities hold
\begin{eqnarray*}
c_l\|{\bf x}\|^2\leq V({\bf x})\leq c_u \|{\bf x}\|^2\\
\dot{V}({\bf x})\leq -c_d \|{\bf x}\|^2\\
\left\|\triangledown V({\bf x})\right\|\leq c_p\|{\bf x}\|.
\end{eqnarray*}
 \hfill{$\square$}\label{asump:e-stability}
\end{assumption}

We first consider the finite Taylor series for ${\bf x}(t, \epsilon)$ in terms of $\epsilon:$
\begin{equation}
{\bf x}(t,\epsilon)={\bf x}^{(0)}(t) +\epsilon {\bf x}^{(1)}(t)+ {\bf e}(t),\label{eq:taylor}
\end{equation} and
\begin{equation}
{\bf x}(0,\epsilon)={\bf x}+\epsilon {\bf z},
\end{equation} where $${\bf x}^{(0)}(t)={\bf x}(t,0) \hbox{ and } {\bf x}^{(1)}(t)=\left.\frac{d {\bf x}}{d \epsilon}(t,\epsilon)\right|_{\epsilon=0}.$$
Substituting (\ref{eq:taylor}) into the dynamical system equation, we get
\begin{eqnarray}
\dot{\bf x}(t,\epsilon)&=&\dot{\bf x}^{(0)}(t) +\epsilon \dot{\bf x}^{(1)}(t)+ \dot{\bf e}(t)=f({\bf x}(t,\epsilon))\\
&=& h^{(0)}({\bf x}^{(0)}(t))+h^{(1)}({\bf x}^{(\leq 1)}(t))\epsilon+R_{\bf e}(t,\epsilon),
\end{eqnarray} where ${\bf x}^{(\leq 1)}=\left({\bf x}^{(0)}, {\bf x}^{(1)}\right).$ The zero-order term $h^{(0)}$ is given by
$$\dot{x}^{(0)}(t)=h^{(0)}\left({\bf x}^{(0)}(t)\right)=f\left({\bf x}^{(0)}(t)\right)\quad \hbox{with}\quad{\bf x}^{(0)}(0)={\bf x},$$ which is the nominal system without the perturbation on the initial condition. The first-order term is given by
\begin{eqnarray*}
h^{(1)}\left({\bf x}^{({\leq 1})}(t)\right)&=&\left.\frac{d}{d \epsilon}f({\bf x}(t,\epsilon))\right|_{\epsilon=0}\\
&=&\left. \frac{\partial f}{\partial x}({\bf x}(t,\epsilon))\frac{d {\bf x}}{d \epsilon}(t,\epsilon)\right|_{\epsilon=0}\\
&=&\frac{\partial f}{\partial x}({\bf x}^{(0)}(t)){\bf x}^{(1)}(t).
\end{eqnarray*} Recall that $\frac{\partial f}{\partial x}$ is the Jacobian matrix. Therefore, we have
\begin{equation}\dot{\bf x}^{(1)}(t)=\frac{\partial f}{\partial x}({\bf x}^{(0)}(t)){\bf x}^{(1)}(t)\quad \hbox{with}\quad{\bf x}^{(1)}(0)={\bf z}. \label{1st-sys}
\end{equation}

We next study ${\bf e}(t)={\bf x}(t,\epsilon)-{\bf x}^{(0)}(t)-\epsilon{\bf x}^{(1)}(t).$ Combining the results above, we have
\begin{eqnarray*}
\dot{\bf e}(t)&=&f\left({\bf x}(t,\epsilon)\right)-f\left({\bf x}^{(0)}(t)\right)-\epsilon\frac{\partial f}{\partial x}({\bf x}^{(0)}(t)){\bf x}^{(1)}(t)\\
{\bf e}(0)&=&0.
\end{eqnarray*}
Now by defining
\begin{eqnarray*}
\rho\left({\bf x}^{(\leq 1)}(t),{\bf e}(t),\epsilon\right)&=&f\left({\bf e}(t)+{\bf x}^{(0)}(t)+\epsilon{\bf x}^{(1)}(t)\right)-f\left({\bf x}^{(0)}(t)+\epsilon{\bf x}^{(1)}(t)\right)-\frac{\partial f}{\partial x}({\bf x}^{(0)}(t)){\bf e}(t)\\
\gamma\left({\bf x}^{(\leq 1)}(t),\epsilon\right)&=&f\left({\bf x}^{(0)}(t)+\epsilon{\bf x}^{(1)}(t)\right)-f\left({\bf x}^{(0)}(t)\right)-\epsilon\frac{\partial f}{\partial x}({\bf x}^{(0)}(t)){\bf x}^{(1)}(t),
\end{eqnarray*}
we obtain
\begin{eqnarray}
\dot{\bf e}(t)&=&\frac{\partial f}{\partial x}({\bf x}^{(0)}(t)){\bf e}(t)+\rho\left({\bf x}^{(\leq 1)}(t),{\bf e}(t),\epsilon\right)+\gamma\left({\bf x}^{(\leq 1)}(t),\epsilon\right).\label{eq: errorpde}
\end{eqnarray}

Note that both $\rho$ and $\gamma$ are $n$-dimensional vectors. It is easy to see that
\begin{equation}
\rho\left({\bf x}^{(\leq 1)}(t),0,\epsilon\right)={\bf 0}.\label{rho-0}
\end{equation} According to Taylor's theorem and the mean value theorem, we have
\begin{eqnarray*}
\gamma_l\left({\bf x}^{(\leq 1)}(t),\epsilon\right)&=&\epsilon^2 {\bf x}^{(1)}(t)^T {\bf H}(f_l)(\xi){\bf x}^{(1)}(t)=\epsilon^2\sum_{i,j}\frac{\partial^2 f_l}{\partial x_i \partial x_j}(\xi)\left(x^{(1)}_i(t)x^{(1)}_j(t)\right)
\end{eqnarray*} for $\xi={\bf x}^{(0)}(t)+\alpha \epsilon {\bf x}^{(1)}(t)$ for some $0\leq \alpha\leq 1.$
${\bf H}(f_l)$ is the Hessian matrix of function $f_l({\bf x}).$
Then we have
\begin{eqnarray*}
\left|\gamma_l\left({\bf x}^{(\leq 1)}(t),\epsilon\right)\right|&=&\epsilon^2\left|\sum_{i,j}\frac{\partial^2 f_l}{\partial x_i \partial x_j}(\xi)\left(x^{(1)}_i(t)x^{(1)}_j(t)\right)\right|.
\end{eqnarray*}

Furthermore, we have
\begin{eqnarray*}
\frac{\partial \rho_l}{\partial e_i}\left({\bf x}^{(\leq 1)}(t),{\bf e}(t),\epsilon\right)&=&\frac{\partial f_l}{\partial x_i}({\bf e}(t)+{\bf x}^{(0)}(t)+\epsilon{\bf x}^{(1)}(t))-\frac{\partial f_l}{\partial x_i}({\bf x}^{(0)}(t)).
\end{eqnarray*}
According to the mean-value theorem and (\ref{rho-0}), we have that
\begin{eqnarray*}
\rho\left({\bf x}^{(\leq 1)}(t),{\bf e}(t),\epsilon\right)= \left(\frac{\partial f}{\partial x}(\tilde{\bf e}(t)+{\bf x}^{(0)}(t)+\epsilon{\bf x}^{(1)}(t))-\frac{\partial f}{\partial x}({\bf x}^{(0)}(t))\right) {\bf e}(t),
\end{eqnarray*}  where $\tilde{\bf e}(t)=a {\bf e}(t)$ for some $0\leq a \leq 1.$
According to the Lipschitz condition in Assumption (\ref{asump:2-diff}) and the Cauchy - Schwarz inequality,  we have
\begin{eqnarray*}
\left\|\rho\left({\bf x}^{(\leq 1)}(t),{\bf e}(t),\epsilon\right)\right\|&\leq& L \left(\|{\bf e}(t)\|+\epsilon\|{\bf x}^{(1)}(t)\|\right)\|{\bf e}(t)\|.
\end{eqnarray*}

Now we utilize the assumption that the nominal system (\ref{ori-sys}) converges to the equilibrium point exponentially fast from any initial condition in the domain. We use the Lyapunov function in Assumption (\ref{asump:e-stability}) to bound $\|{\bf e}(t)\|.$ We start from
\begin{eqnarray*}
&&\dot{V}({\bf e}(t))\\
&=& \triangledown V({\bf e}(t))\cdot \dot{\bf e}(t)\\
&=&\triangledown V({\bf e}(t))\cdot f({\bf e}(t))+\triangledown V({\bf e}(t))\cdot \left(\dot{\bf e}(t)-f({\bf e}(t))\right)\\
&\leq_{(a)}& -c_d{V}({\bf e}(t))+\triangledown V({\bf e}(t))\cdot \left(\dot{\bf e}(t)-f({\bf e}(t))\right)\\
&=& -c_d{V}({\bf e}(t))+\triangledown V({\bf e}(t))\cdot \left(\frac{\partial f}{\partial {\bf x}}({\bf x}^{(0)}(t)){\bf e}(t)-\frac{\partial f}{\partial {\bf x}}(0){\bf e}(t)+\frac{\partial f}{\partial {\bf x}}(0){\bf e}(t)-f({\bf e}(t))\right.\\
&&\hspace{1.7in}\left.+\rho\left({\bf x}^{(\leq 1)}(t),{\bf e}(t),\epsilon\right)+\gamma\left({\bf x}^{(\leq 1)}(t),\epsilon\right)\right)\\
&\leq&-c_d{V}({\bf e}(t))+\left\|\triangledown V({\bf e}(t))\right\|\left(\left\|\frac{\partial f}{\partial {\bf x}}({\bf x}^{(0)}(t)){\bf e}(t)-\frac{\partial f}{\partial {\bf x}}(0){\bf e}(t)\right\|+\left\|\frac{\partial f}{\partial {\bf x}}(0){\bf e}(t)-f({\bf e}(t))\right\|\right.\\
&&\hspace{1.7in}\left.+\left\|\rho\left({\bf x}^{(\leq 1)}(t),{\bf e}(t),\epsilon\right)\right\|+\left\|\gamma\left({\bf x}^{(\leq 1)}(t),\epsilon\right)\right\|\right)
\end{eqnarray*} where inequality (a) is due to assumption (\ref{asump:e-stability}) and the last inequality is a result of the  Cauchy – Schwarz inequality. Note that based on Assumption (\ref{asump:2-diff}) and the mean-value theorem, we have
\begin{eqnarray*}
\left\|\frac{\partial f}{\partial {\bf x}}({\bf x}^{(0)}(t)){\bf e}(t)-\frac{\partial f}{\partial {\bf x}}(0){\bf e}(t)\right\|&\leq& L \|{\bf x}^{(0)}(t)\|\|{\bf e}(t)\|\\
\left\|\frac{\partial f}{\partial {\bf x}}(0){\bf e}(t)-f({\bf e}(t))\right\|&\leq& L\|{\bf e}(t)\|^2.
\end{eqnarray*} We also know that
\begin{eqnarray*}
\left\|\rho\left({\bf x}^{(\leq 1)}(t),{\bf e}(t),\epsilon\right)\right\|&\leq& L \left(\|{\bf e}(t)\|+\epsilon\|{\bf x}^{(1)}(t)\|\right)\|{\bf e}(t)\|\\
\left\|\gamma\left({\bf x}^{(\leq 1)}(t),\epsilon\right)\right\|&=&\epsilon^2A(t),
\end{eqnarray*} where we define
$$A(t)=\sqrt{\sum_l \left(\sum_{i,j}\frac{\partial^2 f_l}{\partial x_i \partial x_j}(\xi)\left(x^{(1)}_i(t)x^{(1)}_j(t)\right)\right)^2}$$ to simplify the notation.

Summarizing the results above, we get
\begin{eqnarray*}
&&\dot{V}({\bf e}(t))\\
&\leq&-c_d{V}({\bf e}(t))+L\left\|\triangledown V({\bf e}(t))\right\|\left(\|{\bf x}^{(0)}(t)\|+\epsilon\|{\bf x}^{(1)}(t)\|+2\|{\bf e}(t)\|\right)\|{\bf e}(t)\|+\left\|\triangledown V({\bf e}(t))\right\|A(t)\epsilon^2\\
&\leq&-c_d{V}({\bf e}(t))+Lc_p\left(\|{\bf x}^{(0)}(t)\|+\epsilon\|{\bf x}^{(1)}(t)\|+2\|{\bf e}(t)\|\right)\|{\bf e}(t)\|^2+c_pA(t)\epsilon^2 \|{\bf e}(t)\|\\
&\leq&-c_d{V}({\bf e}(t))+L\frac{c_p}{c_l}\left(\|{\bf x}^{(0)}(t)\|+\epsilon\|{\bf x}^{(1)}(t)\|+2\|{\bf e}(t)\|\right){V}({\bf e}(t))+\frac{c_p}{\sqrt{c_l}}A(t)\epsilon^2 \sqrt{{V}({\bf e}(t))}.
\end{eqnarray*} Define $W(t)=\sqrt{V(t)},$ then we have
\begin{eqnarray*}
\dot{W}({\bf e}(t))&\leq&-\frac{c_d}{2}{W}({\bf e}(t))+\frac{L}{2}\frac{c_p}{c_l}\left(\|{\bf x}^{(0)}(t)\|+\epsilon\|{\bf x}^{(1)}(t)\|+2\|{\bf e}(t)\|\right){W}({\bf e}(t))+\frac{c_p}{\sqrt{c_1}}A(t)\epsilon^2.
\end{eqnarray*} By the comparison lemma in  \cite{Kha_01}, we have
\begin{eqnarray}
W(t)&\leq&\phi(t,0)W(0)+\frac{c_p}{\sqrt{c_1}}\epsilon^2\int_0^t \phi(t,\tau)A(\tau)\,d\tau\\
&=&\frac{c_p}{\sqrt{c_1}}\epsilon^2\int_0^t \phi(t,\tau)A(\tau)\,d\tau \label{eq:w}
\end{eqnarray} where the transition function $\phi(t,\tau)$ is
$$\phi(t,\tau)=\exp\left(-\frac{c_d}{2}(t-\tau)+\frac{L}{2}\frac{c_p}{c_l}\int_\tau^t\left(\|{\bf x}^{(0)}(\gamma)\|+\epsilon\|{\bf x}^{(1)}(\gamma)\|+2\|{\bf e}(\gamma)\|\right)\,d\gamma \right),$$ and the equality holds because ${\bf e}(0)=0.$

The following lemma proves that the first-order system (\ref{1st-sys}) converges exponentially starting from any initial condition in $\cal D.$
\begin{lemma}
The first-order system (\ref{1st-sys}) is exponentially stable for any solution ${\bf x}^{(0)}(t)$ that starts from $\cal D.$
\end{lemma}
\begin{proof}
That the nominal system is exponentially stable implies that the following linear system
$$\dot{\bf x}=\frac{\partial f}{\partial {x}}(0){\bf x}$$ is (globally) exponentially stable, and $\frac{\partial f}{\partial { x}}(0)$ is Hurwitz (Corollary 4.3 in \cite{Kha_01}), which further implies that there exists a positive definite symmetric matrix ${\bf P}$ such that
$$V({\bf x})={\bf x}^T{\bf P}{\bf x}$$ is a Lyapunov function for the linear system such that
\begin{equation}\dot{V}({\bf x})\leq -\|{\bf x}\|^2.\label{drift}\end{equation}
We start from
\begin{eqnarray*}
\dot{V}({\bf x}^{(1)}(t))&=& \triangledown V({\bf x}^{(1)}(t))\cdot \dot{\bf x}^{(1)}(t)\\
&=&\triangledown V({\bf x}^{(1)}(t))\cdot \frac{\partial f}{\partial {x}}(0){\bf x}^{(1)}(t)+\triangledown V({\bf x}^{(1)}(t))\cdot \left(\frac{\partial f}{\partial {x}}({\bf x}^{(0)}(t)){\bf x}^{(1)}(t)-\frac{\partial f}{\partial {x}}(0){\bf x}^{(1)}(t)\right)\\
&\leq_{(a)}& -\left\|{\bf x}^{(1)}(t)\right\|^2+2\lambda_{\max}\left({\bf P}\right)L\left\|{\bf x}^{(0)}(t)\right\|\left\|{\bf x}^{(1)}(t)\right\|^2\\
&\leq& -\frac{1}{\lambda_{\max}\left({\bf P}\right)}{V}({\bf x}^{(1)}(t))+\frac{2\lambda_{\max}\left({\bf P}\right)}{\lambda_{\min}\left({\bf P}\right)}L\left\|{\bf x}^{(0)}(t)\right\|V\left({\bf x}^{(1)}(t)\right)\\
&\leq& -\left(\frac{1}{\lambda_{\max}\left({\bf P}\right)}-\frac{2\lambda_{\max}\left({\bf P}\right)}{\lambda_{\min}\left({\bf P}\right)}L\left\|{\bf x}^{(0)}(t)\right\|\right)V\left({\bf x}^{(1)}(t)\right)
\end{eqnarray*} where inequality $(a$) is based on (\ref{drift}) and the definition of $V({\bf x})$, Assumption (\ref{asump:2-diff}) and the mean-value theorem, and $\lambda_{\max}({\bf P})$ is the largest eigenvalue of matrix $\bf P.$

By the comparison lemma, we have
\begin{eqnarray*}
V(t)&\leq& \exp\left(-\frac{1}{\lambda_{\max}({\bf P})}t+\frac{2\lambda_{\max}\left({\bf P}\right)}{\lambda_{\min}\left({\bf P}\right)}L\int_0^t \left\|{\bf x}^{(0)}(\tau)\right\| \,d\tau\right)V(0)\\
&\leq& \exp\left(\frac{2\lambda_{\max}\left({\bf P}\right)L\kappa \|{\bf x}(0)\|}{\alpha \lambda_{\min}\left({\bf P}\right)}\right)\exp\left(-\frac{1}{\lambda_{\max}({\bf P})}t\right)V(0),
\end{eqnarray*} where the last inequality holds because the exponential convergence assumption (\ref{asump:e-stability}) yields
\begin{eqnarray*}
\int_0^t \left\|{\bf x}^{(0)}(\tau)\right\| \,d\tau\leq\int_0^{\infty} \left\|{\bf x}^{(0)}(\tau)\right\| \,d\tau \leq \frac{\kappa \|{\bf x}(0)\|}{\alpha}.
\end{eqnarray*}
Recall that ${\bf x}^{(1)}(0)={\bf z},$ so $V(0)={\bf z}^T {\bf P}{\bf z}=p$ and
\begin{eqnarray}
\|{\bf x}^{(1)}(t)\|^2&\leq&  \frac{p}{\lambda_{\min}({\bf P})}\exp\left(\frac{2\lambda_{\max}\left({\bf P}\right)L\kappa \|{\bf x}(0)\|}{\alpha \lambda_{\min}\left({\bf P}\right)}\right)\exp\left(-\frac{1}{\lambda_{\max}({\bf P})}t\right).\label{eq:exconv-1}
\end{eqnarray}
\end{proof}

From the lemma above and assumption (\ref{asump:2-diff}), we have that there exists a constant $\mu$ such that
$$A(t)\leq \mu \|{\bf x}^{(1)}(t)\|^2\leq \frac{ \mu p}{\lambda_{\min}({\bf P})}\exp\left(\frac{2\lambda_{\max}\left({\bf P}\right)L\kappa \|{\bf x}(0)\|}{\alpha \lambda_{\min}({\bf P})}\right)\exp\left(-\frac{1}{\lambda_{\max}({\bf P})}t\right).$$

Consider the set such that $\|{\bf e}(t)\|\leq \frac{c_d c_l}{4c_p L},$ we have
\begin{eqnarray}
&&\phi(t,\tau)\\
&\leq&\exp\left(-\frac{c_d}{2}(t-\tau)+\frac{L}{2}\frac{c_p}{c_l}\int_\tau^t\left(\|{\bf x}^{(0)}(\gamma)\|+\epsilon\|{\bf x}^{(1)}(\gamma)\|\right) \,d\gamma \right)\\
&\leq&\exp\left(-\frac{c_d}{4}(t-\tau)+\frac{L}{2}\frac{c_p}{c_l}\left(\frac{\kappa \|{\bf x}(0)\|}{\alpha}+2\epsilon \frac{\lambda_{\max}({\bf P})\sqrt{p}}{\sqrt{\lambda_{\min}({\bf P})}} \exp\left(\frac{\lambda_{\max}({\bf P})L\kappa \|{\bf x}(0)\|}{\alpha \lambda_{\min}({\bf P})}\right)\right)\right),
\end{eqnarray} where last inequality yields from (\ref{eq:exconv}) and  (\ref{eq:exconv-1}).

Recall we have inequality (\ref{eq:w})
\begin{eqnarray}
W(t)&\leq&\frac{c_p}{\sqrt{c_1}}\epsilon^2\int_0^t \phi(t,\tau)A(\tau)\,d\tau \label{eq:w2}.
\end{eqnarray}  Substituting the bounds on $\phi(t,\tau)$ and $A(\tau),$ we obtain
\begin{eqnarray*}
W(t)&\leq& \epsilon^2 \sqrt{c_1}Z(\|{\bf x}(0)\|)\int_0^t \exp\left(-\frac{c_d}{4}(t-\tau)-\frac{1}{\lambda_{\max}({\bf P})}\tau\right)\,d\tau,
\end{eqnarray*} where
\begin{eqnarray*}
&&Z(\|{\bf x}(0)\|)\\
&=&\frac{c_p}{{c_1}}\exp\left(\frac{L}{2}\frac{c_p}{c_l}\left(\frac{\kappa \|{\bf x}(0)\|}{\alpha}+2\epsilon \frac{\lambda_{\max}({\bf P})\sqrt{p}}{\sqrt{\lambda_{\min}({\bf P})}} \exp\left(\frac{\lambda_{\max}({\bf P})L\kappa \|{\bf x}(0)\|}{\alpha \lambda_{\min}({\bf P})}\right)\right)\right)\times\\
&&\frac{ \mu p}{\lambda_{\min}({\bf P})}\exp\left(\frac{2\lambda_{\max}\left({\bf P}\right)L\kappa \|{\bf x}(0)\|}{\alpha \lambda_{\min}({\bf P})}\right)\\
&=&\frac{c_p \mu p}{{c_1}\lambda_{\min}({\bf P})}\exp\left(\frac{\kappa L}{\alpha}\left(\frac{c_p}{2c_l}+\frac{2\lambda_{\max}({\bf P})}{\lambda_{\min}({\bf P})}\right) \|{\bf x}(0)\|+\epsilon\sqrt{p}\frac{c_p L}{c_l} \frac{\lambda_{\max}({\bf P})}{\sqrt{\lambda_{\min}({\bf P})}} \exp\left(\frac{\lambda_{\max}({\bf P})L\kappa \|{\bf x}(0)\|}{\alpha \lambda_{\min}({\bf P})}\right)\right).
\end{eqnarray*} In other words, we have
\begin{eqnarray}
\|{\bf e}(t)\|&\leq& \epsilon^2 Z(\|{\bf x}(0)\|)\int_0^t \exp\left(-\frac{c_d}{4}(t-\tau)-\frac{1}{\lambda_{\max}({\bf P})}\tau\right)\,d\tau\\
&=&\left\{
     \begin{array}{ll}
       \epsilon^2 Z(\|{\bf x}(0)\|)\frac{1}{\frac{c_d}{4}-\frac{1}{\lambda_{\max}({\bf P})}}\left(\exp\left(-\frac{1}{\lambda_{\max}({\bf P})}t\right)-\exp\left(-\frac{c_d}{4}t\right)\right), & \hbox{ if } \frac{c_d}{4}\not=\frac{1}{\lambda_{\max}({\bf P})} \\
       \epsilon^2 Z(\|{\bf x}(0)\|)t \exp\left(-\frac{c_d}{4}t\right) , & \hbox{ otherwise.}
     \end{array}
   \right.
.\label{eq:error}
\end{eqnarray} Then
\begin{eqnarray*}
\int_0^\infty \|{\bf e}(t)\|\,dt&\leq& \left\{
     \begin{array}{ll}
       \epsilon^2 Z(\|{\bf x}(0)\|)\frac{4 \lambda_{\max}({\bf P})}{c_d} & \hbox{ if } \frac{c_d}{4}\not=\frac{1}{\lambda_{\max}({\bf P})}\\
       \epsilon^2 Z(\|{\bf x}(0)\|)\frac{16}{c_d^2} , & \hbox{ otherwise.}
     \end{array}
   \right..\\
&=&\epsilon^2 Z(\|{\bf x}(0)\|)\frac{4 \lambda_{\max}({\bf P})}{c_d}.
\end{eqnarray*} It is easy to see that with properly defined $\alpha_1,$ $\alpha_2,$ $\alpha_3$ and $\alpha_4,$ we have
\begin{eqnarray}
\int_0^\infty \|{\bf e}(t)\|\,dt&\leq& \epsilon^2 p \alpha_1\exp\left(\alpha_2\frac{\|{\bf x}(0)\|}{\alpha}+\epsilon \sqrt{p} \alpha_3\exp\left(\alpha_4 \frac{\|{\bf x}(0)\|}{\alpha}\right)\right).\label{eq:accumulated-error}
\end{eqnarray} We keep the terms $\|{\bf x}(0)\|$ and $\alpha$ to show that the cumulative error depends on the initial condition and the convergence rate of the mean-field model. Furthermore, $p={\bf z}^T {\bf P}{\bf z}\leq \lambda_{\max}({\bf P})\|{\bf z}\|^2.$

\section{Conclusion}
This paper studies the convergence of the stationary distributions of a family of CTMCs to the mean-field limit. When the mean-field model is perfect, the rate of convergence (the mean-square difference) has been proved to be $O(1/M).$ Based on Stein's method for bounding the distance of probability distributions and the perturbation theory for nonlinear systems, a fundamental connection between the convergence to the mean-field limit and the stability of the mean-field model has been established.

\section*{Acknowledgement} The author is very grateful to Jim Dai and Anton Braverman. Jim's seminar on Stein's method for the steady-state diffusion approximations inspired this work. The discussions with Jim and Anton had continuously stimulated the author during the writing of this paper.
\bibliographystyle{IEEEtran}
\bibliography{U:/bib/inlab-refs}

\begin{thebibliography}{10}
\providecommand{\url}[1]{#1}
\csname url@samestyle\endcsname
\providecommand{\newblock}{\relax}
\providecommand{\bibinfo}[2]{#2}
\providecommand{\BIBentrySTDinterwordspacing}{\spaceskip=0pt\relax}
\providecommand{\BIBentryALTinterwordstretchfactor}{4}
\providecommand{\BIBentryALTinterwordspacing}{\spaceskip=\fontdimen2\font plus
\BIBentryALTinterwordstretchfactor\fontdimen3\font minus
  \fontdimen4\font\relax}
\providecommand{\BIBforeignlanguage}[2]{{%
\expandafter\ifx\csname l@#1\endcsname\relax
\typeout{** WARNING: IEEEtran.bst: No hyphenation pattern has been}%
\typeout{** loaded for the language `#1'. Using the pattern for}%
\typeout{** the default language instead.}%
\else
\language=\csname l@#1\endcsname
\fi
#2}}
\providecommand{\BIBdecl}{\relax}
\BIBdecl

\bibitem{Kad_09}
L.~P. Kadanoff, ``More is the same; phase transitions and mean field
  theories,'' \emph{Journal of Statistical Physics}, vol. 137, no. 5-6, pp.
  777--797, 2009.

\bibitem{Bai_75}
N.~T.~J. Bailey, \emph{The mathematical theory of infectious diseases and its
  applications}.\hskip 1em plus 0.5em minus 0.4em\relax Hafner Press, 1975.

\bibitem{BarKarKel_92}
F.~Baccelli, F.~Karpelevich, M.~Y. Kelbert, A.~Puhalskii, A.~Rybko, and Y.~M.
  Suhov, ``A mean-field limit for a class of queueing networks,'' \emph{Journal
  of statistical physics}, vol.~66, no. 3-4, pp. 803--825, 1992.

\bibitem{VveDobKar_96}
N.~D. Vvedenskaya, R.~L. Dobrushin, and F.~I. Karpelevich, ``Queueing system
  with selection of the shortest of two queues: An asymptotic approach,''
  \emph{Problemy Peredachi Informatsii}, vol.~32, no.~1, pp. 20--34, 1996.

\bibitem{Mit_96}
M.~Mitzenmacher, ``The power of two choices in randomized load balancing,''
  Ph.D. dissertation, University of California at Berkeley, 1996.

\bibitem{LasLio_07}
J.-M. Lasry and P.-L. Lions, ``Mean field games,'' \emph{Japanese Journal of
  Mathematics}, vol.~2, no.~1, pp. 229--260, 2007.

\bibitem{Kur_71}
T.~G. Kurtz, ``Limit theorems for sequences of jump markov processes
  approximating ordinary differential processes,'' \emph{J. Appl. Probab.},
  vol.~8, no.~2, pp. 344--356, 1971.

\bibitem{Kur_81}
------, \emph{Approximation of population processes}.\hskip 1em plus 0.5em
  minus 0.4em\relax SIAM, 1981, vol.~36.

\bibitem{YinSriKan_15}
L.~Ying, R.~Srikant, and X.~Kang, ``The power of slightly more than one sample
  in randomized load balancing,'' in \emph{Proc. IEEE Int. Conf. Computer
  Communications (INFOCOM)}, Hong Kong, 2015.

\bibitem{Szn_91}
A.-S. Sznitman, ``Topics in propagation of chaos,'' in \emph{Ecole
  d'{\'e}t{\'e} de probabilit{\'e}s de Saint-Flour XIX—1989}, 1991, pp.
  165--251.

\bibitem{AnaBen_93}
V.~Anantharam and M.~Benchekroun, ``A technique for computing sojourn times in
  large networks of interacting queues,'' \emph{Probability in the engineering
  and informational sciences}, vol.~7, no.~04, pp. 441--464, 1993.

\bibitem{BraLuPra_12}
M.~Bramson, Y.~Lu, and B.~Prabhakar, ``Asymptotic independence of queues under
  randomized load balancing,'' \emph{Queueing Systems}, vol.~71, no.~3, pp.
  247--292, 2012.

\bibitem{MukMaz_13}
A.~Mukhopadhyay and R.~R. Mazumdar, ``Analysis of load balancing in large
  heterogeneous processor sharing systems,'' \emph{arXiv preprint
  arXiv:1311.5806}, 2013.

\bibitem{Ste_72}
C.~Stein, ``A bound for the error in the normal approximation to the
  distribution of a sum of dependent random variables,'' in \emph{Proc. Sixth
  Berkeley Symp. Math. Stat. Prob.}, 1972, pp. 583--602.

\bibitem{Ste_86}
------, ``Approximate computation of expectations,'' \emph{Lecture
  Notes-Monograph Series}, vol.~7, pp. i--164, 1986.

\bibitem{BraDai_15}
A.~Braverman and J.~Dai, ``Stein's method for steady-state diffusion
  approximations of $ m/ph/n+ m $ systems,'' \emph{arXiv preprint
  arXiv:1503.00774}, 2015.

\bibitem{Gur_14}
I.~Gurvich \emph{et~al.}, ``Diffusion models and steady-state approximations
  for exponentially ergodic markovian queues,'' \emph{Adv. in Appl. Probab.},
  vol.~24, no.~6, pp. 2527--2559, 2014.

\bibitem{Bar_88}
A.~D. Barbour, ``Stein's method and {Poisson} process convergence,'' \emph{J.
  Appl. Probab.}, pp. 175--184, 1988.

\bibitem{Got_91}
F.~Gotze, ``On the rate of convergence in the multivariate clt,'' \emph{Ann.
  Probab.}, pp. 724--739, 1991.

\bibitem{Kha_01}
H.~K. Khalil, \emph{Nonlinear systems}.\hskip 1em plus 0.5em minus 0.4em\relax
  Prentice Hall, 2001.

\end{thebibliography}
\end{document}